\documentclass[11pt]{article}
\usepackage{ifthen,amsmath,amssymb,amsthm,latexsym,color,epsfig,bm,verbatim,hyperref}%,yhmath,pseudocode,fancybox}
\usepackage{multirow}
\usepackage{newtype}

\usepackage{pstricks}
\usepackage{pst-tree}
\usepackage{pst-node}
\usepackage{pst-plot}
\usepackage{wrapfig}

\newcommand{\GF}[1]{{\mathrm{GF}(#1)}}

\newcommand{\RN}{\mathbb{R}}

\newcommand{\ER}{\ensuremath{\bm{\exists \RN}}}

\newtype{\class}{\mathbf}
\class{L} \class{P} \class{PH} \class{NP} \class{coNP}  \class{PSPACE}

\newtheorem{theorem}{Theorem}[section]
\newtheorem{lemma}[theorem]{Lemma}

\newtheorem{corollary}[theorem]{Corollary}

\theoremstyle{definition}

\newtheorem{claim}{Claim}
\newtheorem{remark}[claim]{Remark}
\makeatletter
\@addtoreset{claim}{theorem}
\makeatother

\newcommand{\SAT}{\textsf{\upshape SAT}}
\newcommand{\STRETCH}{\textsf{\upshape STRETCHABILITY}}

\psset{unit=1pt}

\begin{document}

\title{Picking Planar Edges; or, Drawing a Graph with a Planar Subgraph}

\author{
{Marcus Schaefer
} \\
{\small School of Computing} \\[-0.13cm]
{\small DePaul University} \\[-0.13cm]
{\small Chicago, Illinois 60604, USA} \\[-0.13cm]
{\small \tt mschaefer@cdm.depaul.edu}\\[-0.13cm]
}

\maketitle

\begin{abstract}
 Given a graph $G$ and a subset $F \subseteq E(G)$ of its edges, is there a drawing of $G$ in which all edges of $F$ are free of crossings? We show that this question can be solved in polynomial time using a Hanani-Tutte style approach. If we require the drawing of $G$ to be straight-line, and allow at most one crossing along each edge in $F$, the problem turns out to be as hard as the existential theory of the real numbers.
\end{abstract}

\section{Introduction}

Angelini, Binucci, Da Lozzo, Didimo, Grilli, Montecchiani, Patrignani, and Tollis~\cite{ABDLDGMPT13a}
asked the following problem:
\begin{quote}
``Given a non-planar graph $G$ and a planar subgraph $S$ of $G$, decide whether $G$ admits a drawing $\Gamma$
such that the edges of $S$ are not crossed in $\Gamma$, and compute $\Gamma$ if it exists''.
\end{quote}
Their paper studies two variants of this problem: the unrestricted problem in which $\Gamma$ is an arbitrary poly-line drawing, and the straight-line variant, in which $\Gamma$ is restricted to straight-line drawings. Let us call these the {\em partial planarity} and the {\em geometric partial planarity} problem. It seems that these two problems are new to the literature, Maybe the closest previous variant is the (also very recent) notion of partially embedded planarity~\cite{ABFJKPR10}, which differs in that a particular embedding of $S$ is given, and the desired planar embedding of $G$ has to extend the given embedding of $S$. For partially embedded planarity, a linear-time testing algorithm is known~\cite{ABFJKPR10}, as well as an obstruction set~\cite{JKR13}.

David Eppstein commented in his blog~\cite{E13b}:
\begin{quote}
``If you're given a graph in which some edges are allowed to participate in crossings while others must remain uncrossed, how can you draw it, respecting these constraints? Unfortunately the authors were unable to determine the computational complexity of this problem, and leave it as an interesting open problem''.
\end{quote}
In other words, given a graph $G$ and a subset of its edges $F \subseteq E(G)$, is there a (straight-line) drawing of $G$ in which all edges of $F$ are free of crossings? The two formulations are equivalent, of course, but we slightly prefer the second, since it shows that we can specify for each edge whether it has to be planar (crossing-free) or not; we can pick the planar edges. Looking at planarity as a local requirement opens it up for combination with other properties; for example, what happens if we can specify a bound on the number of crossings along each edge, or on the number of bends? 

%We will look at the special case of geometric partial $1$-planarity, in which we are allowed to limit the number of crossings to a single crossing.

\subsection*{Previous Research}
Angelini, et al.~\cite{ABDLDGMPT13a} show that $(G,S)$ is always partially planar if $S$ is a spanning tree of $G$, even if the embedding of $S$ is required to be a straight-line embedding. For geometric partial planarity, they show that $(G,S)$ can always be realized if $S$ is a spanning spider or caterpillar, even in polynomial area. However, they also exhibit examples of $(G,S)$ where $S$ is a spanning tree of $G$ for which $(G,S)$ has no geometric partial realization. There are further algorithms in the paper to test geometric partial planarity for various types of spanning trees $S$, though in some cases the layout algorithms require exponential area.

\subsection*{Our Contribution}
In Section~\ref{sec:PPHT} we show that using a Hanani-Tutte style approach successfully settles the complexity of the poly-line variant of the problem: partial planarity can be solved in polynomial time. This is a further example of a planarity-style problem in which there is no traditional polynomial-time algorithm for the problem, but the Hanani-Tutte approach leads to a solution. Other examples of this are surveyed in~\cite{S13a}.

We have to leave the complexity of the straight-line variant open, but there is a good chance that it is as hard as the existential theory of the reals (see~\cite{S10b}). One indication for this is that the layout algorithm for geometric partial planarity suggested in~\cite{ABDLDGMPT13a} needs exponential area on some inputs. Secondly, the result is true if we replace planarity with $1$-planarity: partial geometric $1$-planarity is as hard as the existential theory of the reals, as we will see in Section~\ref{sec:GP1P}. In comparison, the special case of geometric $1$-planarity is \NP-complete (this follows from known results in the literature, see Theorem~\ref{thm:g1pNP}). 

\section{Partial Planarity and Hanani-Tutte}\label{sec:PPHT}

We assume that the reader is somewhat familiar with the Hanani-Tutte characterization of planarity (see~\cite{S13a,S14}). Briefly, Hanani~\cite{C34} and Tutte~\cite{T70} established the following algebraic characterization of planar graphs: a graph is planar if and only if it has a drawing in which every two independent edges cross evenly. This criterion can be rephrased as a linear system over $\GF{2}$: Create variables $x_{e,v}$ for every $e \in E(G)$ and $v \in V(G)$, and let $i_D(e,f)$ denote the number of times two edges $e$ and $f$ cross in a drawing $D$ of $G$.
Fix an arbitrary drawing $D$ of $G$ (e.g.\ a convex drawing).
Let $P(D)$ be the following system over $\GF{2}$:
\[i_{D}(uv, st) + x_{uv,s} + x_{uv,t} + x_{st, u} + x_{st,v} \equiv 0 \bmod{2},\]
for every pair of independent edges $uv, st \in E(G)$. Then $G$ is planar if and only if $P(D)$ is solvable.
The heart of the proof is showing that solvability of $P(D)$ leads to a drawing of $G$, we will not explain this part
(see~\cite[Section 3]{S13a} for a detailed discussion). The other direction is a consequence of the following well-known fact about drawings: as far as the crossing parity between pairs of independent edges is concerned, one can turn any drawing of a graph into any other drawing of the graph by performing a set of $(e,v)$-moves, where an {\em $(e,v)$-move} consists of taking a small piece of $e$, moving it close to $v$ and then pushing it over $v$; the effect of an $(e,v)$-move is that the crossing parity between $e$ and any edge incident to $v$ changes. Imagining one drawing of a graph morphing into another, it is easy to believe that $(e,v)$-moves are sufficient to get from one drawing to another. We state this result without proof. For further details see~\cite[Section 4.6]{dL13} or~\cite[Lemma 1.12]{S14}.

\begin{lemma}\label{lem:cob}
 If $D$ and $D'$ are two drawings of the same graph $G$, then there is a set of $(e,v)$-moves so that
 \[i_{D'}(uv,st) \equiv i_{D}(uv,st) + x_{uv,s} + x_{uv,t} + x_{st, u} + x_{st,v} \bmod{2},\]
 for all edges $uv$, $st \in E(G)$, where $x_{e,v} = 1$ if an $(e,v)$-move is performed, and $x_{e,v} = 0$ otherwise.
\end{lemma}

For a graph $G$ with a set of edges $F \subseteq E(G)$, fix an arbitrary drawing $D$ of $G$, and let $P(D,F)$ be the following system of equations over $\GF{2}$:
\[i_{D}(uv, st) + x_{uv,s} + x_{uv,t} + x_{st, u} + x_{st,v} \equiv 0 \bmod{2},\]
for every pair of independent edges $uv \in F$ and $st \in E(G)$.

\begin{lemma}\label{lem:PGS}
 $G$ has a drawing $\Gamma$ in which $F$ is free of crossings if and only if $P(D,F)$ is solvable.
\end{lemma}

Since the solvability of a linear system of equations over a field (in this case $\GF{2}$) can be decided in polynomial time, the following corollary is immediate.

\begin{corollary}
  Given a graph $G$ with a set of edges $F \subseteq E(G)$, it can be decided in polynomial time whether $G$ has a drawing in which all edges in $F$ are free of crossings.
\end{corollary}

The running time of the algorithm is on the order $O((nm)^3)$, where $n = |V(G)|$ and $m = |E(G)|$, since systems of linear equations over a field can be solved in cubic time, and $P(D,F)$ can have as many as $O(nm)$ equations and $O(nm)$ variables (note that we can assume that $|F| = O(n)$: if the graph $(V(G),F)$ is not planar, then there is no drawing of $G$ in which all edges of $F$ are free of crossings; on the other hand, we cannot assume that $G$ is planar). This may seem impractical at a first glance, but recent experiments with an algorithm of this type have been quite successful~\cite{GMS1?}.

The hard direction in the proof of Lemma~\ref{lem:PGS} is covered by the following result from an earlier paper on the independent odd crossing number~\cite{PSS10}. We call an edge $e$ in a drawing $D$ {\em independently even} if it crosses every edge independent of it an even number of times. More formally, $i_D(e,f) \equiv 0 \bmod{2}$ for every $f$ which is independent of $e$.

\begin{lemma}[Pelsmajer, Schaefer, \v{S}tefankovi\v{c}~\cite{PSS10}]\label{lem:PTindeven}
 If $D$ is a drawing of a graph $G$ in the plane, then
 $G$ has a drawing in which the independently even edges of $D$ are crossing-free
 and every pair of edges crosses at most once.
\end{lemma}

The proof of Lemma~\ref{lem:PTindeven} is constructive in the sense that the new drawing of $G$ can be found in polynomial time (there are no explicit time bounds, but a running time quadratic in $O(|G|)$ seems achievable).

\begin{proof}[Proof of Lemma~\ref{lem:PGS}]
 Suppose $P(D,F)$ is solvable, and fix a solution $x_{e,v} \in \{0,1\}$, for $e \in E(G), v \in V(G)$, for some initial drawing $D$ of $G$. Construct a drawing $D'$ from $D$ by performing an $(e,v)$-move for every $e \in E(G)$ and $v \in V(G)$ for which $x_{e,v} = 1$. Pick $uv \in F$ and let $st \in E(G)$ be an arbitrary edge independent of $uv$. Then
 \[i_{D'}(uv,st) = i_D(uv,st) + x_{uv,s} + x_{uv,t} + x_{st, u} + x_{st,v} \equiv 0 \bmod{2},\]
 since $x_{e,v}$ is a solution of the system $P(D,F)$. Thus $uv$ is independently even. Since $uv$ was arbitrary, all edges in $F$ are independently even, and, by Lemma~\ref{lem:PTindeven}, there is a drawing of $G$ in which all edges of $F$ are free of crossings, which is what we had to show.

 For the other direction, assume $G$ has a drawing $D'$ in which all edges of $F$ are free of crossings. By Lemma~\ref{lem:cob} we know that there is a set of $(e,v)$-moves so that
 \[i_{D'}(uv,st) \equiv i_D(uv,st) +  x_{uv,s} + x_{uv,t} + x_{st, u} + x_{st,v} \bmod{2}\]
 for all pairs of independent edges $uv$, $st \in E(G)$. Now if $uv \in F$, then $i_{D'}(uv,st) = 0$ for every edge $st \in E(G)$. In particular,
 \[i_D(uv,st) +  x_{uv,s} + x_{uv,t} + x_{st, u} + x_{st,v} \equiv i_{D'}(uv,st) \equiv 0 \bmod{2},\]
 so $x_{e,v}$ is a solution to $P(D,F)$, which is what we had to show.
\end{proof}

\section{Geometric Partial $1$-Planarity}\label{sec:GP1P}

In the straight-line version of the partial planarity problem, we ask whether for a given $G$ and $F \subseteq E(G)$, there is a straight-line drawing of $G$ in which the edges of $F$ are free of crossings. We cannot settle the complexity of this problem, but we have a suggestive result for a generalized version. Suppose we are allowed to specify sets $F_k \subseteq E(G)$, and ask whether $G$ has a straight-line drawing in which all edges in $F_k$ have at most $k$ crossings, for every $k$. The problem posed by Angelini, Binucci, Da Lozzo, Didimo, Grilli, Montecchiani, Patrignani, and Tollis~\cite{ABDLDGMPT13a} corresponds to specifying a set $F_0$ of crossing-free edges. We will show that if instead we specify a set $F_1$ of edges that may be crossed at most once, the problem has the same complexity as deciding the truth of statements in the existential theory of the reals; in the terminology introduced in~\cite{SS09, S10b}, it is {\em \ER-complete}. In analogy with the notion of $1$-planarity (in which every edge may be crossed at most once), we call the problem {\em geometric partial $1$-planarity}.

\begin{remark}[Equivalent Drawings]
Geometric $1$-planarity was first studied by Eggleton~\cite{E86} and Thomassen~\cite{T88}, and more recently in~\cite{HELP12}, and several other papers, but with one important difference: in these papers one is given an intial $1$-planar drawing of $G$ and asks whether there is an equivalent geometric $1$-planar drawing, where two drawings are {\em equivalent} if they have the same facial structure (for this definition to make sense, we consider crossings to be vertices). With this stronger notion, Thomassen~\cite{T88} was able to identify forbidden subconfigurations, which led to a linear-time testing algorithm~\cite{HELP12}. Similarly, Nagamochi~\cite{N13} shows that if we are given a drawing of $G$ and a $2$-connected, spanning subgraph $S$ of $G$, one can test in linear time whether there is an equivalent drawing of $G$ in which edges of $S$ are free of crossings.
\end{remark}

We will not give a formal definition of \ER\ and \ER-completeness (that can be found in~\cite{SS09, S10b}), instead we will work with \STRETCH, a complete problem for the class. This is just like working with \SAT\ (or any other \NP-complete problem) rather than the formal class \NP.

An {\em arrangement of pseudolines} in the plane is a collection of $x$-monotone curves (that is, each pseudoline has exactly one crossing with every vertical line) so that every pair of pseudolines crosses exactly once. An arrangement of pseudolines is {\em stretchable} if all pseudolines can be replaced by straight lines so that the order of crossings along the lines remains the same. See Figure~\ref{fig:arrs} for an example of a pseudoline arrangement, and an equivalent straight-line arrangement.

\begin{figure}[htb]
\begin{center}
\begin{tabular}{cp{0.3in}c}
\begin{pspicture}(0,0)(100,100)

\pscurve(0,90)(30,30)(70,10)(100,10)
\pscurve(5,60)(60,40)(90,30)
\pscurve(10,30)(20,35)(50,70)(70,70)(95,60)
\pscurve(5,10)(80,80)(100,90)
\end{pspicture}

&&

\begin{pspicture}(0,0)(100,100)
\pnode(1,96.6){a1}
\pnode(60,20){a5}
\pnode(5.3,83.3){a2}
\pnode(79,45){a6}
\pnode(12.6,63.3){a3}
\pnode(90,70){a7}
\pnode(20,46.6){a4}
\pnode(95.6,86.6){a8}

\ncline{a1}{a5}
\ncline{a2}{a6}
\ncline{a3}{a7}
\ncline{a4}{a8}
\end{pspicture}

\end{tabular}
\end{center}
\caption{{\em (Left)} A pseudoline arrangement. {\em (Right)} A straight-line arrangement equivalent to the pseudoline arrangement on the left.}\label{fig:arrs}
\end{figure}
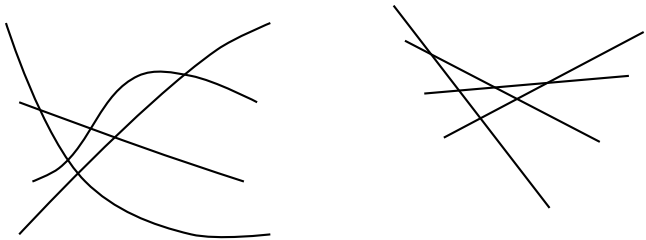

Mn\"{e}v~\cite{M88} showed that \STRETCH, the problem of deciding whether an arrangement of pseudolines is stretchable, is computationally equivalent to deciding the truth of a sentence in the existential theory of the real numbers (for an accessible treatment of Mn\"{e}v's proof, see Shor~\cite{S91}).\footnote{Mn\"{e}v actually showed a stronger result, his universality theorem, here we are only interested in the computational aspects.} This led to the introduction of the complexity class \ER, which contains all problems which can be translated in polynomial time to a sentence in the existential theory of the reals, see~\cite{SS09, S10b} for more details. Similar to the theory of \NP-completeness, there are \ER-complete problems including stretchability, and truth in the existential theory of the reals, but many other problems as well, such as the rectilinear crossing number (there is a wikipedia page, for example~\cite{W12}). We note that \ER\ contains \NP, since the existential theory of the real numbers easily encodes satisfiability, and in turn \ER\ is contained in \PSPACE, due to a famous result by Canny~\cite{C88}. Therefor any \ER-complete problem, such as partial geometric $1$-planarity is \NP-hard, and can be solved in \PSPACE.

\begin{theorem}\label{thm:PG1P}
 Partial Geometric $1$-Planarity is \ER-complete.
\end{theorem}

In particular, we conclude that the problem is \NP-hard, and lies in \PSPACE. For the proof we make use of a simple gadget.

\begin{lemma}\label{lem:K6}
  There is no drawing of a $K_6$ and a vertex-disjoint cycle $C$ so that all edges in $K_6$ have at most one crossing, and there is a crossing between an edge of $K_6$ and the cycle.
\end{lemma}

\begin{proof}[Proof of Lemma~\ref{lem:K6}]
 Suppose there were a drawing as described in the lemma, in which a $K_6$-edge $e = uv$ crosses a cycle edge $f \in E(C)$. 
 Then $e$ cannot cross any of the edges in $E(C) - \{f\}$, since it has at most one crossing, and thus no edge incident to $u$ can cross an edge incident to $v$: to have a common point, one of them would have to cross $C$, but then it would have two crossings, one with the cycle, and one with the other edge. Therefore, the edges adjacent to $e$ do not cross each other at all. This implies that the drawing of the $K_6$ contains $4$ triangles with a shared edge $e$ whose other edges do not cross each other. On the sphere, there is only one such drawing: $4$ nested triangles (with a common base). But this implies that two of the endpoints of those triangles are separated by the other two triangles, which means the original endpoints cannot be joined by an edge in a $1$-planar drawing of the $K_6$, since it would have to cross the other two triangles (it cannot cross $e$, since $e$ already has a crossing).
\end{proof}

\begin{proof}[Proof of Theorem~\ref{thm:PG1P}]
 The problem can easily be expressed using an existentially quantified statement over the real numbers: use the existential quantifiers to find the locations of the vertices of the graph; once the vertices are located, it is easy to express that each edge in $F$ is crossed at most once. This shows that the problem lies in \ER.

 Since stretchability of pseudoline arrangements is \ER-complete, it is sufficient to show that stretchability can be reduced to partial geometric $1$-planarity to establish \ER-hardness of partial geometric $1$-planarity. Let ${\cal A}$ be an arbitrary arrangement of pseudolines. We construct a graph $G_{\cal A}$ and a set of edges $F \subseteq G_{\cal A}$ so that ${\cal A}$ is stretchable if and only if $G_{\cal A}$ has a straight-line drawing in which every edge in $F$ has at most one crossing.

 Let $R$ be a parabola-shaped region (boundary of the form $y = x^2 + c$ for some constant $c \in \RN$)
 so that all crossings of pseudolines in ${\cal A}$ lie within the region $R$. Let $V_{\cal A}$ be the intersection points of pseudolines with the parabolic boundary of $R$ (we can assume that all crossings of pseudolines lie in the convex hull of $V_1$). The region $R$ is separated by ${\cal A}$ into faces, some of them adjacent to the boundary of $R$, and some of them inner faces of the arrangement. 
 We choose a set of vertices $V_I$ consisting of an interior vertex for each inner face of the arrangement; for each face on the boundary of $R$, we pick a vertex on the interior of a boundary arc of the face, let $V_B$ of those boundary vertices; note that all faces except the infinite face, are incident to a unique boundary arc; the infinite face is incident to two boundary arcs, of which we pick one arbitrarily to place the $V_B$-vertex. Finally, pick a vertex $p$ below $R$ so that $p$ can {\em see} all vertices of $V_{\cal A} \cup V_B$; that is, a straight-line segment between $p$ and any vertex in $V_{\cal A} \cup V_B$ does not cross the boundary of $R$. Let $V = V_{\cal A} \cup V_R \cup V_I \cup \{p\}$.
 
 For every two vertices in $V_{\cal A}$ belonging to the same pseudoline, add an edge between those vertices. Add a {\em frame} as follows: connect the vertices of $V_{\cal A} \cup V_B$ by a cycle that corresponds to the order of those vertices along the boundary of $R$, and connect $p$ to every vertex in $V_{\cal A} \cup V_B$ by an edge. Identify each edge of the frame with an edge in a (new) $K_6$. Finally, add the dual graph of the line arrangement to $V_I \cup V_B$. Let $F$ consist of all edges, except for the edges corresponding to the original pseudolines. See Figure~\ref{fig:ex} for an example.
 
 \begin{figure}[htb]
\begin{center}
\begin{pspicture}(0,-10)(300,300)

\parabola[linecolor=gray](0,300)(150,50)

\put(255,30){\small
\begin{tabular}{ll}
$\bullet$ & $V_{\cal A}$ \\
$\circ$ & $V_B$ \\
$\blacktriangle$ & $V_I$ \\
{\footnotesize $\blacksquare$} & $p$ \\[0.1in]
\begin{pspicture}(0,0)(10,5)\psline(0,2.5)(10,2.5) \end{pspicture} & ${\cal A}$ \\
\begin{pspicture}(0,0)(10,5)\psline[linestyle=dashed](0,2.5)(10,2.5) \end{pspicture} & dual of ${\cal A}$ \\
\begin{pspicture}(0,0)(10,5)\psline[linestyle=dotted](0,2.5)(10,2.5) \end{pspicture} & frame \\
\begin{pspicture}(0,0)(10,5)\psline[linecolor=gray](0,2.5)(10,2.5) \end{pspicture} & $\partial R$ 
\end{tabular}}

\cnode*(3,290){2pt}{a1}
\cnode*(180,60){2pt}{a5}
\cnode*(15.5,250){2pt}{a2}
\cnode*(237.5,135){2pt}{a6}
\cnode*(38,190){2pt}{a3}
\cnode*(270,210){2pt}{a7}
\cnode*(60,140){2pt}{a4}
\cnode*(287.5,260){2pt}{a8}

\ncline[linewidth=1.4]{a1}{a5}
\ncline[linewidth=1.4]{a2}{a6}
\ncline[linewidth=1.4]{a3}{a7}
\ncline[linewidth=1.4]{a4}{a8}

\cnode(9.5,270){2pt}{b1}
\cnode(26.2,220){2pt}{b2}
\cnode(48.2,165){2pt}{b3}
\cnode(82.7,100){2pt}{b4}
\cnode(201.8,80){2pt}{b5}
\cnode(247,155){2pt}{b6}
\cnode(279,235){2pt}{b7}
\cnode(297,290){2pt}{b8}

\dotnode*[dotstyle=triangle, dotscale=1.6](82,205){i1}
\dotnode*[dotstyle=triangle, dotscale=1.6](114,180){i2}
\dotnode*[dotstyle=triangle, dotscale=1.6](145,192){i3}

\dotnode*[dotstyle=square, dotscale=1.6](150,0){p}\nput{-45}{p}{$p$}

\ncline[linestyle=dashed]{b2}{i1}
\ncline[linestyle=dashed]{b3}{i2}
\ncline[linestyle=dashed]{b4}{b5}
\ncline[linestyle=dashed]{b5}{i2}
\ncline[linestyle=dashed]{b6}{i3}
\ncline[linestyle=dashed]{b8}{i1}
\ncline[linestyle=dashed]{b8}{i3}
\ncline[linestyle=dashed]{b8}{b1}
\ncline[linestyle=dashed]{i1}{i2}
\ncline[linestyle=dashed]{i1}{i3}
\ncline[linestyle=dashed]{i2}{i3}

\nccurve[linestyle=dotted, ncurv=0.05]{a1}{b1}
\nccurve[linestyle=dotted, ncurv=0.05]{b1}{a2}
\nccurve[linestyle=dotted, ncurv=0.05]{a2}{b2}
\nccurve[linestyle=dotted, ncurv=0.05]{b2}{a3}
\nccurve[linestyle=dotted, ncurv=0.05]{a3}{b3}
\nccurve[linestyle=dotted, ncurv=0.05]{b3}{a4}
\nccurve[linestyle=dotted, ncurv=0.05]{a4}{b4}
\nccurve[linestyle=dotted, ncurv=0]{b4}{a5}
\nccurve[linestyle=dotted, ncurv=-0.2]{a5}{b5}
\nccurve[linestyle=dotted, ncurv=-0.05]{b5}{a6}
\nccurve[linestyle=dotted, ncurv=-0.3]{a6}{b6}
\nccurve[linestyle=dotted, ncurv=-0.1]{b6}{a7}
\nccurve[linestyle=dotted, ncurv=-0.3]{a7}{b7}
\nccurve[linestyle=dotted, ncurv=-0.3]{b7}{a8}
\nccurve[linestyle=dotted, ncurv=-0.2]{a8}{b8}
\nccurve[linestyle=dotted, ncurv=0]{b8}{a1}

\nccurve[linestyle=dashed, ncurv=0.2, angleA=-30,angleB=80]{b1}{b2}
\nccurve[linestyle=dashed, ncurv=0.2, angleA=-30,angleB=80]{b2}{b3}
\nccurve[linestyle=dashed, ncurv=0.2, angleA=-30,angleB=80]{b3}{b4}
\nccurve[linestyle=dashed, ncurv=0.2,angleA=90,angleB=-140]{b5}{b6}
\nccurve[linestyle=dashed, ncurv=0.2,angleA=90,angleB=-140]{b6}{b7}
\nccurve[linestyle=dashed, ncurv=0.2,angleA=90,angleB=-140]{b7}{b8}

\nccurve[linestyle=dotted, angleA = 170, angleB=-90]{p}{a1}
\nccurve[linestyle=dotted, angleA = 170, angleB=-90]{p}{a2}
\nccurve[linestyle=dotted, angleA = 170, angleB=-90]{p}{a3}
\nccurve[linestyle=dotted, angleA = 170, angleB=-90]{p}{a4}
\nccurve[linestyle=dotted, angleA = 10, angleB=-90]{p}{a5}
\nccurve[linestyle=dotted, angleA = 10, angleB=-90]{p}{a6}
\nccurve[linestyle=dotted, angleA = 10, angleB=-90]{p}{a7}
\nccurve[linestyle=dotted, angleA = 10, angleB=-90]{p}{a8}
\nccurve[linestyle=dotted, angleA = 170, angleB=-90]{p}{b1}
\nccurve[linestyle=dotted, angleA = 170, angleB=-90]{p}{b2}
\nccurve[linestyle=dotted, angleA = 170, angleB=-90]{p}{b3}
\nccurve[linestyle=dotted, angleA = 170, angleB=-90]{p}{b4}
\nccurve[linestyle=dotted, angleA = 10, angleB=-90]{p}{b5}
\nccurve[linestyle=dotted, angleA = 10, angleB=-90]{p}{b6}
\nccurve[linestyle=dotted, angleA = 10, angleB=-90]{p}{b7}
\nccurve[linestyle=dotted, angleA = 10, angleB=-90]{p}{b8}
\end{pspicture}
\end{center}
\caption{The graph $G_{\cal A}$ corresponding to the pseudoline arrangement ${\cal A}$ shown in Figure~\ref{fig:arrs}. $K_6$-gadgets are not shown, and some edges are curved to improve readability.}\label{fig:ex}
\end{figure}
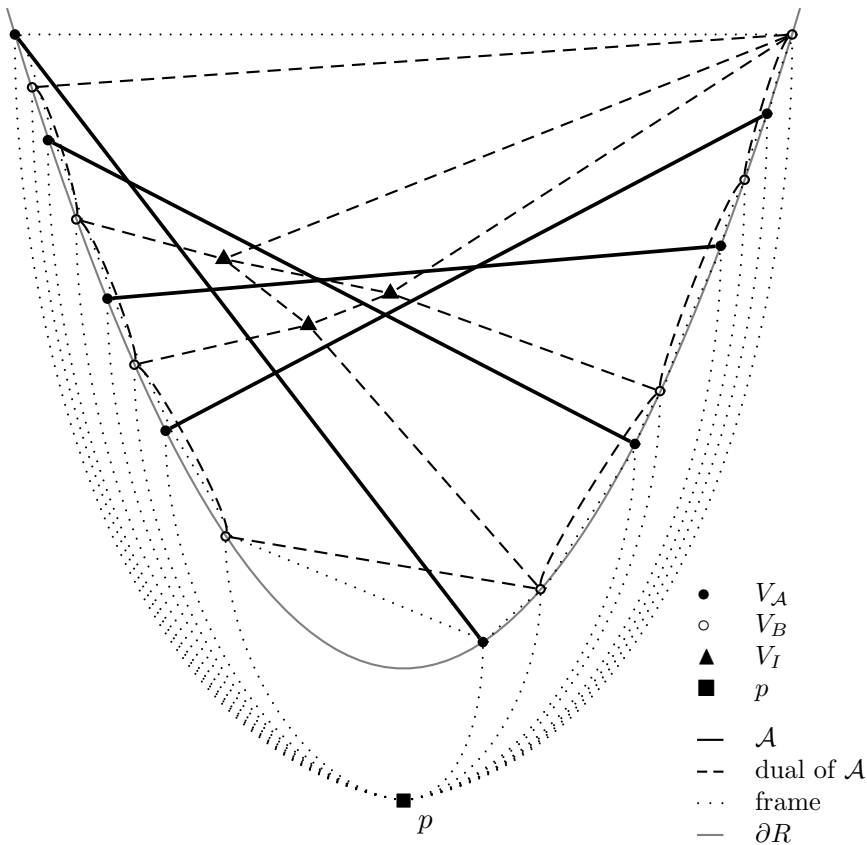
 
 We first note that if ${\cal A}$ is stretchable, then $G_{\cal A}$ has a straight-line drawing in which every edge of $F$ has at most one crossing. To see this, start with a straight-line realization of ${\cal A}$. Perform the construction of $G_{\cal A}$ as we described it above. Because of the convexity of $R$, we can draw the edges of the cycle on $V_{\cal A} \cup V_B$ as well as the straight-line edges to $p$. We can then add a straight-line drawing of each $K_6$ gadget to the frame so that the shared edge is free of crossings (and the remainder of $K_6$ does not participate in unnecessary crossings). Finally, the dual graph of the line arrangement can be added to $V_I \cup V_B$ since any edge connects two vertices in adjacent faces of the line arrangement which is always possible with a straight-line arrangement, unless the resulting edge coincides with the boundary of a cell. This cannot occur, however, since $V_I$ vertices lie in the interior of faces, and the $V_B$ vertices lie on the boundary of the convex region $R$. In this drawing, every edge in $F$ has at most one crossing. Only edges corresponding to the original pseudolines are crossed more than once by dual edges.
 
 For the other direction, start with a straight-line drawing of $G_{\cal A}$ in which all edges in $F$ have at most one crossing. Suppose $f$ is an edge of the frame and let $e$ be another edge in $G_{\cal A}$ which does not belong to $f$'s $K_6$ gadget. If $e$ and $f$ are adjacent, they cannot cross, since the drawing is straight-line. Hence $e$ either belongs to another $K_6$-gadget or is one of the edges between vertices in $V_{\cal A} \cup V_B \cup V_I$. In either case, $e$ belongs to a cycle which is vertex-disjoint from $f$'s $K_6$-gadget, so Lemma~\ref{lem:K6}
 implies that $e$ does not cross $f$. This means that after removal of all the $K_6$-gadgets, the frame is free of crossings. In particular, the cycle $C$ on $V_{\cal A} \cup V_B$ is embedded, and hence its vertices occur in the order determined by the line arrangement ${\cal A}$. Let ${\cal A'}$ be the line arrangement obtained from $G_{\cal A}$ by erasing the frame (and its gadgets), the dual graph, and extending the edges corresponding to pseudolines to infinite lines. We claim that ${\cal A'}$ is equivalent to ${\cal A}$.
 
 We just saw that the order of pseudolines along $C$ is correct, and, since the frame does not cross edges corresponding to pseudolines, every two such edges have to cross inside the region bounded by $C$ (since their endpoints along $C$ alternate in ${\cal A'}$ just as they do in ${\cal A}$. We now show that the dual graph of ${\cal A}$ forces the facial structure of the line arrangement to be unique.
  
 Let $v \in V_I \cup V_B$ be an arbitrary vertex representing a face of the line arrangement, and $e$ an edge corresponding to some line in ${\cal A}$. We show that $v$ lies on the same side of $e$ (within the region bounded by the cycle $C$ through $V_{\cal A} \cup V_B$) in both ${\cal A}$ and ${\cal A'}$, so the two line arrangements have to be equivalent. If $v \in V_B$ this is forced by the cycle $C$; if $v \in V_I$, we argue as follows: let $s$ and $t$ be the $V_B$-vertices closest to $e$ (along $C$) and on the same side of $e$ as $v$. We claim that there is an $st$-path of length $|{\cal A}|-1$ on $V_I$ vertices. To see this, start at $s$. Since $v$ is an inner vertex, $s$ and $v$ do not lie in the same face of the line arrangement, hence there must be an edge $f$ corresponding to a line of ${\cal A}$ so that $s$ and $v$ lie on opposite sides of $f$. More strongly, there must be such an edge $f$ which contributes to the boundary of the cell $v$ lies in (if the two vertices were on the same side of all lines contributing to the boundary of the cell, they would have to be in the same cell); in other words, there is a cell adjacent to the cell containing $v$ (sharing $f$), which is closer to $s$ (note that $t$ and $v$ have to lie on the same side of $f$, since otherwise $s$ and $t$ lie on the same side of both $e$ and $f$, but then they cannot both be closest to $e$). By induction we can now show that there are paths $sv$ and $vt$ containing at most $|{\cal A}|-1$ edges together (since $e$ need never be crossed). But then the path $svt$ in $G_{\cal A}$ on $|{\cal A}|-1$ edges cannot cross $e$, since it has to cross all $|{\cal A}|-1$ edges corresponding to pseudolines (other than $e$). Hence $v$ lies on the same side of $e$ in both ${\cal A}$ and ${\cal A'}$. 
 
 Since ${\cal A'}$ is a straight-line arrangement equivalent to ${\cal A}$, we conclude that ${\cal A}$ is stretchable, which is what we had to show.
\end{proof}

We have to leave open the question whether geometric partial planarity is \ER-hard as well, but as we mentioned earlier, we can show that geometric $1$-planarity is only \NP-complete. This follows from two well-known results: $1$-planarity is \NP-complete~\cite{BG07,KM08,CM13}, and geometric $1$-planarity can be tested in linear time if a rotation system is given~\cite{T88,HELP12}.

\begin{theorem}[Folklore]\label{thm:g1pNP}
 Testing geometric $1$-planarity is \NP-complete.
\end{theorem}
\begin{proof}
 The problem lies in \NP, since we can guess the rotation system, and then use the linear time algorithm from~\cite{HELP12} to check whether there is an obstruction to geometric $1$-planarity with that rotation system. To see \NP-hardness, we use that it is \NP-hard to test $1$-planarity. If a graph $G$ is $1$-planar, then it has a $1$-planar drawing in which each edge has at most one bend: simply apply Fary's theorem to the graph obtained from $G$ by replacing each crossing by a dummy vertex. To avoid that crossings and bends occur at the same location, we replace each edge in $G$ with a path of length three to get a new graph $G'$. Then $G$ is $1$-planar if and only if $G'$ has a geometric $1$-planar embedding in which all edges incident to the original vertices of $G$ are free of crossings. And that we can easily guarantee by identifying all of these edges with an edge of a $K_6$-gadget. Let $H$ be this new graph. Then $G$ is $1$-planar if and only if $H$ is geometrically $1$-planar. Therefore, geometric $1$-planarity is \NP-hard.
\end{proof}

\section{Future Research}

What can we say about traditional approaches to partial planarity? More specifically, can $PQ$-trees or $SPQR$-trees be used to solve this problem?

Recall that a {\em bridge} of $S$ in $G$ is either an edge in $E(G)-E(S)$ with both endpoints in $S$ (a {\em trivial} bridge) or a connected component of $G-S$ together with its edges and vertices of attachment to $S$. Given an embedding of $S$, a group of vertices of $S$ is {\em mutually visible}~\cite{ABFJKPR10} if there is a face of $S$ containing all vertices in the group on its boundary. The poly-line variant can be rephrased as follows: is there a poly-line embedding of $S$ so that for every bridge of $S$ in $G$, the vertices of attachment of the bridge are mutually visible? It seems quite likely that SPQR-trees could be used to decide that question, even in linear time, extending ideas for deciding partially embedded planarity developed in~\cite{ABFJKPR10}. 

Another solution may come from progress on simultaneous embeddings, since partial planarity can be rephrased as a special case of the simultaneous embedding with fixed edges problem (the sunflower case).\footnote{This connection was pointed out to me by Ignaz Rutter.} A family of graphs $(G_i)_{i=1}^n$, has a {\em simultaneous embedding (with fixed  edges)} if there is a drawing of $\cup_{i=1}^n G_i$ in which no two edges belonging to the same graph cross each other. In the {\em sunflower case}, $G_i \cap G_j = H$ for some fixed graph $H$ and all $i < j$. Given $G$ and $S$, we let $G_i$ be $S$ together with a star whose vertices of attachment to $S$ are the same as the vertices of attachment of the $i$-th $G$-bridge of $S$. In a simultaneous embedding of $(G_i)_{i=1}^n$, each star forces the vertices of attachment of the $i$-th bridge to lie in the same face of $S$. Hence, any algorithmic solution of the sunflower case, solves the corresponding partial planarity problem. At this point, this is only known for $S$ being a disjoint union of components which are $2$-connected, or have max-degree at most $3$ (via a Hanani-Tutte argument~\cite{S13a}), but there maybe be more traditional PQ-tree algorithms for this case in the future.

\bibliographystyle{plain}
\bibliography{pcr}

\end{document}